\documentclass[reqno,tbtags,intlimits,a4paper,oneside,12pt]{amsart}
\usepackage[cp1251]{inputenc}%
\usepackage[english]{babel}
\usepackage{amssymb,upref,calc,url}
\usepackage[mathscr]{eucal}
\usepackage{graphicx}
\usepackage{amsmath,amscd,amsthm,verbatim}
\usepackage[all]{xy}
\usepackage[in]{fullpage}
\newtheorem{thm}{\hskip\parindent Theorem}[section]
\newtheorem{lem}[thm]{\hskip\parindent Lemma}
\newtheorem{prb}[thm]{\hskip\parindent Problem}
\newtheorem{cor}[thm]{\hskip\parindent Corollary}

\theoremstyle{definition}
\newtheorem{dfn}{\hskip\parindent Definition}[section]
\newtheorem{ex}[dfn]{\hskip\parindent Example}
\newtheorem{rem}[dfn]{\hskip\parindent Remark}
\DeclareMathOperator{\Der}{Der}
\DeclareMathOperator{\wt}{wt}

\begin{document}

\title{Parametric Korteweg--de Vries hierarchy and~hyperelliptic sigma functions \footnotemark[1]}


\author{E.\,\,Yu.~Bunkova}
\address{Steklov Mathematical Institute of Russian Academy of Sciences, Moscow, Russia}
\email{bunkova@mi-ras.ru}

\author{V.\,\,M.~Buchstaber}
\address{Steklov Mathematical Institute of Russian Academy of Sciences, Moscow, Russia}
\email{buchstab@mi-ras.ru}

\markboth{Parametric Korteweg--de Vries hierarchy}{E.\,\,Yu.~Bunkova, V.\,\,M.~Buchstaber}

\maketitle


\begin{abstract}
In this paper we define the parametric Korteweg-de Vries hierarchy that depends on an infinite set of graded parameters $a = (a_4,a_6,\dots)$. We show that, for any genus $g$, the Klein hyperelliptic function $\wp_{1,1}(t,\lambda)$ defined on the basis of the multidimensional sigma function $\sigma(t, \lambda)$, where $t = (t_1, t_3,\dots, t_{2g-1})$, $\lambda = (\lambda_4, \lambda_6,\dots, \lambda_{4 g + 2})$, determines a solution of this hierarchy, where the parameters $a$ are given as polynomials in the parameters $\lambda$ of the sigma function.

The proof uses results on the family of operators introduced by V.~M.~Buchstaber and S.~Yu.~Shorina. This family consists of $g$ third-order differential operators of $g$ variables. Such families are defined for all $g \geqslant 1$, the operators in each of them commute in pairs and also commute with the Schr\"odinger operator.

In this paper, we describe the relationship between these families and the parametric Korteweg--de Vries hierarchy. A similar infinite family of third-order operators on an infinite set of variables is constructed. The results obtained are extended to the case of~such a family.
\end{abstract}

\footnotetext[1]{This research was supported by the Russian Science Foundation (grant \textnumero 20-11-19998), \text{https://rscf.ru/project/20-11-19998/}.}

\section{Introduction}\label{S0}

Emma Previato is the author of outstanding results in the theory of Abelian functions and their applications in mathematical physics. The problems discussed in this article were at the center of her scientific interests.

Consider the classical \emph{Korteweg--de Vries equation} \cite{KdVt} of the form
\begin{equation} \label{eqcKdV}
\frac{\partial u}{\partial t} = \frac14 ( u''' - 6 u u'),
\end{equation}
where $u$ is a smooth function of $z$ and $t$, and the prime denotes the derivative with respect to $z$. We consider~$z$ and~$t$ to be complex variables.  In what follows, we will use the same notation for a function and the operator of multiplication by this function.

Equation \eqref{eqcKdV} can be represented in the Lax form \cite{Lax}
\begin{equation} \label{LA3}
\frac{\partial u}{\partial t} = [L, A_3],
\end{equation}
where the second and third order operators
\begin{align}
L &= \partial_z^2 - u, \label{L0}\\
A_3 &= \partial_z^3 - \frac34(u \partial_z + \partial_z u) \notag
\end{align}
are defined for a smooth function $u$ of a complex variable $z$. Here, the operator $-L = - \partial_z^2 + u$ is the \emph{Schr\"odinger operator} with potential $u$.

Let $U_3 = A_3 - \partial_t$. Then, equation \eqref{eqcKdV} takes the form $[L, U_3] = 0$.

For a smooth function $u$ of a complex variable $z$, we denote by $\mathcal{H}$ the graded $\mathbb{C}$-algebra generated by homogeneous polynomials in the variables $u, u', u'',\dots, u^{(k)},\break\dots$, where $\wt u = 2$, $\wt u^{(k)} = k+2$. We denote the arguments $u, u', u'',\dots, u^{(k)},\dots$ of such polynomials by $[u]$. The operators
\begin{equation} \label{Ak}
A_{2k+1} = \partial_z^{2k+1} + \sum_{j=1}^k ( \alpha_{2j}^{2k+1}[u] \partial_z^{2k-2j+1} + \partial_z^{2k-2j+1} \alpha_{2j}^{2k+1}[u]),
\end{equation}
where $\alpha_{2j}^{2k+1}[u] \in \mathcal{H}$, $\wt \alpha_{2j}^{2k+1}[u] = 2j$, depending on $u$, are uniquely determined by the fundamental property that for any smooth function $u$ the condition $[L, A_{2k+1}] \in \mathcal{H}$ holds. This definition is based on results of \cite{Lax} and \cite{BM}; see for details and examples \S\ref{Ssss}. See also \cite{Nov}--\cite{Sok} for various ways of constructing these operators. Let us introduce the notation $r_{2k+3}[u] = [L, A_{2k+1}]$. Then, $r_{2k+3}[u] \in \mathcal{H}$ and $\wt r_{2k+3}[u] = 2k+3$.

Equation~\eqref{LA3} extends to the~\emph{canonical Korteweg--de Vries hierarchy}
\begin{equation} \label{KdVh}
\frac{\partial u}{\partial t_{2k+1}} = [L, A_{2k+1}], \qquad k \in \mathbb{N},
\end{equation}
where $u$ is a smooth function of the complex variable $z = t_1$ and an infinite set of complex variables $t_3, t_5, t_7,\dots$\,.

The \emph{parametric Korteweg--de Vries hierarchy} is the system
\begin{equation} \label{KdVhld}
\frac{\partial u}{\partial t_{2k+1}} = [L, A_{2k+1} + \sum_{j=0}^{k-2} a_{2k-2j} A_{2j+1}], \qquad k \in \mathbb{N},
\end{equation}
where $u$ is a smooth function of the complex variable $z = t_1$ and an infinite set of complex variables $t_3, t_5,\dots, t_{2k+1},\dots$, while we call the set of constants $a_4, a_6,\dots$, $a_{2k},\dots$ the parameters of hierarchy \eqref{KdVhld}. Note that the parameters of the $k$-th equation of this hierarchy coincide with the corresponding parameters of the $j$-th equation of the hierarchy.

For a smooth function $u$ of a complex variable $z$, we denote by $\mathcal{H}_A$ the graded $\mathbb{C}$-algebra generated by homogeneous polynomials in the variables $u, u', u'',\dots , u^{(k)},\dots$ and $a_4, a_6,\dots, a_{2k},\dots$, where~$\wt u = 2$, $\wt u^{(k)} = k+2$, $\wt a_{2k} = 2k$. Then, the right-hand sides of the equations of hierarchy \eqref{KdVhld} lie in $\mathcal{H}_A$.

The $g$-th equation of hierarchy \eqref{KdVhld}, in the case when $\partial u / \partial t_{2g+1} = 0$, is called the \emph{higher stationary $g$-KdV equation}. This equation can be represented as
\begin{equation} \label{gKdV}
r_{2g+3}[u] + \sum_{j=0}^{g-2} a_{2g-2j} r_{2j+3}[u] = 0,
\end{equation}
where $u$ is a smooth function of the complex variable $z$, and we call the set of constants $a_4, a_6,\dots, a_{2g}$ the parameters of the higher stationary $g$-KdV equation \eqref{gKdV} .

Let us show how to define the canonical Korteweg--de Vries hierarchy \eqref{KdVh} and the parametric Korteweg--de Vries hierarchy~\eqref{KdVhld} based on the Lenard operator
\begin{equation} \label{Len}
\mathcal{R} = \frac14 \partial_z^2 - \frac12 u' \partial_z^{-1} - u.
\end{equation}

In the $\mathbb{C}$-algebra $\mathcal{H}$, where the function $u$ is treated as a free variable, the operator $\partial_z$ acts injectively on homogeneous polynomials $h \in \mathcal{H}$ with $ \wt h > 0$. Therefore, the operator $\partial_z^{-1}$ is uniquely defined on the image $\partial_z h$, where $h \in \mathcal{H}$, $\wt h > 0$, by the relation $\partial_z^{-1} (\partial_z h) = h$. According to Lemma 44 in \cite{BM} or the corollary \ref{lastcor} of this paper, there exist differential polynomials $\nu_{2k+2}[u] \in \mathcal{H}$, $\wt \nu_{2k+ 2}[u] = 2k+2$ such that $r_{2k+3}[u] = \partial_z (\nu_{2k+2}[u])$. Thus, the operator $\partial_z^{-1}$ and, hence, the Lenard operator \eqref{Len} are uniquely defined on the polynomials $r_{2k+3}[u]$. According to relation (9) Chap. 2 in \cite{GD} (see also Theorem 1.2 in \cite{BS}) the right-hand sides of the canonical Korteweg--de Vries hierarchy \eqref{KdVh} are related by
\begin{equation} \label{len}
\mathcal{R}(r_{2k+1}[u]) = r_{2k+3}[u],
\end{equation}
which allows us to restore the canonical Korteweg--de Vries hierarchy~\eqref{KdVh}.

In the $\mathbb{C}$-algebra $\mathcal{H}_A$, where the function $u$ is considered as a free variable and $a_4, a_6,\dots,$ $a_{2k},\dots$ are parameters independent of $u$, the action of the operator $\partial_z^{-1}$ on the image of the operator $\partial_z$ on the homogeneous polynomial $h \in \mathcal{H}_A$, $\wt h = 2k$, is defined by $\partial_z^{-1}(\partial_z h) = h + c_{2k}$ up to a graded $u$-independent integration constant $c_{2k} \in \mathcal{H}_A$, $\ wt c_{2k} = 2k$, \text{$\partial_z(c_{2k}) = 0$.} Therefore, for $h$ and $c_{2k}$, as above, according to formula \eqref{Len}, the action of the operator $\mathcal{R}$ is defined up to the term $-\frac12 c_{2k} u'$. The relation \eqref{len} is true in the algebra $\mathcal{H}$. Applying it in the $\mathcal{H}_A$ algebra to the right hand side of the $(k-1)$-th equation of the parametric Korteweg--de Vries hierarchy \eqref{KdVhld}, we obtain
\[
\mathcal{R}\bigg(r_{2k+1}[u] + \sum_{j=0}^{k-3} a_{2k-2j-2} r_{2j+3}[u]\bigg ) = r_{2k+3}[u] + \sum_{j=1}^{k-2} a_{2k-2j} r_{2j+3}[u] - \frac12 c_{2k} r_{3 }[u],
\]
where $r_{3}[u] = u'$. Setting $c_{2k} = - 2 a_{2k}$, we obtain the right hand side of the $k$-th equation of the parametric Korteweg--de Vries hierarchy \eqref{KdVhld}. Thus, the Lenard operator makes it possible to reconstruct the parametric Korteweg--de Vries hierarchy \eqref{KdVhld} up to the parameters $a_4, a_6,\dots, a_{2k},\dots$ arising from the multivaluedness of the Lenard operator.

Let~$\partial_{2k+1} = \partial/\partial t_{2k+1}$ and
\[
U_{2k+1} = A_{2k+1} - \partial_{2k+1}.
\]
If $u$ is a solution to the Korteweg--de Vries hierarchy \eqref{KdVh}, then the second-order operator $L$ and operators $U_{2k+1}$ of order $2k+1$, for $k = 1,2 ,\dots$, commute. In~\cite{BS} and~\cite{BS1}, a family of $g$ commuting third-order multidimensional differential operators was introduced, each of which commutes with the operator $L$. We call them the \emph{Buchstaber--Shorina operators}. In~\cite{BS} it was shown the close connection of these operators with the Korteweg--de Vries hierarchy~\eqref{KdVh} and the higher stationary $g$-KdV equation \eqref{gKdV}. In \S \ref{OBS} we define the Buchstaber--Shorina operators and present the results of papers~\cite{BS}--\cite{BS2}, which we will make essential use of in the current work.

In this paper, we succeeded in obtaining a generalisation of the results from the paper~\cite{BS}. In theorem \ref{fullhd}, we show that the function involved in the definition of the Buchstaber--Shorina operators is a solution of the parametric Korteweg--de Vries hierarchy \eqref{KdVhld} and find all relations between the constants defined by the Buchstaber--Shorina operators and the parameters of the Korteweg--de Vries hierar-\break chy. In theorem \ref{tgreat}, we obtain a generalisation of the key result of the paper~\cite{BS} and the theorem~\ref{fullhd} to the case of an infinite family of Buchstaber--Shorina operators. Such a family is constructed on the basis of the solution of the parametric Korteweg--\break de Vries hierarchy \eqref{KdVhld}.

It was noted in~\cite{BS1} that the Buchstaber--Shorina operators were obtained by studying the properties of \emph{hyperelliptic Klein functions} (see \cite{Baker}--\cite{BEL-12} and \cite{WW} for elliptic functions). We recall definitions of the theory of hyperelliptic functions that are important for us in~\S\ref{Sh}. In~\S\ref{1BS} we demonstrate the applicability of results on Buchstaber--Shorina operators to such functions.

Let $\sigma(t, \lambda)$ be a hyperelliptic sigma function, where the variables $t = (t_1, t_3,\dots,\break t_{2g-1})$ and the parameters $\lambda = (\lambda_4, \lambda_6,\dots, \lambda_{4 g + 2})$ are introduced in~\S \ref{Sh}. Let us set
\begin{equation}
\zeta_{k} = \partial_k \ln \sigma(t, \lambda),\quad
\wp_{k_1,\dots, k_n} = - \partial_{k_1} \cdots \partial_{k_n} \ln \sigma(t, \lambda), \label{pijdef}
\end{equation}
where $n \geqslant 2$ and $k_s \in \{ 1, 3,\dots, 2 g - 1\}$. For a hyperelliptic sigma function $\sigma(t, \lambda)$ of a nondegenerate hyperelliptic curve of genus $g$, the function
\begin{equation} \label{e1}
u = 2 \wp_{1,1} = - 2 \partial_1^2 \ln \sigma(t, \lambda)
\end{equation}
defines a solution of the higher stationary $g$-KdV equation~\eqref{gKdV} with respect to the variable $z=t_1$; see Theorem~4.12 in~\cite{BEL-97} and corollary \ref{cor52} in \S \ref {1BS}. In theorem~4.12 in~\cite{BEL-97}, it is shown that function $u = 2 \wp_{1,1}$, given by~\eqref{e1}, is a solution to the Korteweg--de Vries hierarchy and its definition is based on the Lenard operator \eqref{Len}. In this paper, in Theorem \ref{mainthm} , we prove that function $u = 2 \wp_{1,1}$ given by \eqref{e1} is a solution to the parametric Korteweg--de Vries hierarchy \eqref{KdVhld} , and we explicitly specify all parameters $a_4, a_6,\dots$, $a_{2k},\dots$ of the Korteweg--de Vries \eqref{KdVhld} hierarchy in terms of the parameters $\lambda$ of the sigma function.

In order to prove Theorems \ref{fullhd}, \ref{mainthm} and \ref{tgreat}, we use the construction developed in~\cite{B2} and~\cite{B3} and described in \S\ref{s3}, which relates the universal bundle $\pi\colon\mathcal{U}_g \to \Lambda_g \subset \mathbb{C}^{2g}$, whose fiber is the Jacobian of a nonsingular hyperelliptic curve of genus $g$, with the polynomial mapping $\rho_g \colon \mathbb{C}^{3g}\to \mathbb{C}^{2g}$. This construction corresponds to the transition from the field of hyperelliptic functions to the ring of polynomials from generators of the subfield of this field. We use the results presented in \S\ref{OBS} and \S\ref{1BS} to obtain the structure of the polynomial map $\rho_g$. We prove Theorem 6.5, according to which the set of mappings $\{ \rho_g,\, g=1,2,\dots \}$, consistent with the growth of $g$, defines the limit mapping $\rho\colon \mathbb{C}^{3 \infty}\to \mathbb{C}^{2\infty}$ and allows us to express the operators of our parametric Korteweg--de Vries hierarchy as polynomial operators in the graded coordinates of the space $\mathbb{C}^{3\infty}$. The corresponding results are presented in \S\ref{s8} and \S\ref{Sinf}. Note that the set of mappings $ \pi$ is not consistent with the growth of $g$; see Remark 6.6.

In \S \ref{s8}, we describe the polynomial differential operators obtained from the construc-\break tion in \S \ref{s3} and introduce the \emph{Buchstaber--Shorina polynomial differential operators}. In \S \ref{Sinf}, we generalise these results to the infinite-dimensional case, we introduce the concept of a \emph{polynomial parametric Korteweg--de Vries hierarchy} and obtain the proofs of Theorems \ref{fullhd}, \ref{mainthm} and~\ref{tgreat}.

\section{Definition of operators $A_{2k+1}$} \label{Ssss}

\begin{lem} \label{lemA}
In the $\mathbb{C}$-algebra $\mathcal{H}$, where the function $u$ is treated as a free variable, for each $k \in \mathbb{N}$, there is a unique operator $A_{2k+1} $ of the form \eqref{Ak}, where $\alpha_{2j}^{2k+1}[u] \in \mathcal{H}$ and $\wt \alpha_{2j}^{2k+1}[u] = 2j$, such that $[L, A_{2k+1}] \in \mathcal{H}$.
\end{lem}
\begin{proof}
The proof of the existence of operators $A_{2k+1}$ of this form is given, for example, in \S 6 of the paper~\cite{BM}.

To prove the uniqueness of such operators, we use a method derived from the corresponding proof in \cite{Lax}. For the operators~$L$ and $A_{2k+1}$ given by \eqref{L0} and~\eqref{Ak}, the commutator $[L, A_{2k+1}]$ is a differential operator of order not higher than $2k$, namely,
\[
[L,A_{2k+1}] = r_{2k+3}[u] + \sum_{j=1}^{2k} \beta_{2k+3-j}^{2k+1}[u] \partial^j,
\]
where $\beta_{2k+3-j}^{2k+1}[u] \in \mathcal{H}$, $\wt \beta_{2k+3-j}^{2k+1}[u] = 2k+3-j$. For the $j$-th term in \eqref{Ak} we have
\[
[L, \alpha_{2j}^{2k+1}[u] \partial_z^{2k-2j+1} + \partial_z^{2k-2j+1} \alpha_{2j}^{2k+1}[ u]] = 4 (\alpha_{2j}^{2k+1}[u])' \partial_z^{2k-2j+2} + M_{2j}^{2k+1},
\]
where $M_{2j}^{2k+1}$ is a differential operator of order at most ${2k-2j+1}$. Thus, the condition $\beta_{2s+1}^{2k+1}[u] = 0$, for $s = 1,\dots, k$, gives an expression for~$(\alpha_{2s}^{2k+ 1}[u])'$ in terms of $u, u', u'',\dots$ and $\alpha_{2i}^{2k+1}[u], (\alpha_{2i}^{2k+ 1}[u])', (\alpha_{2i}^{2k+1}[u])'',\dots$, where~$i<s$. For the problem of finding the polynomial $\alpha_{2s}^{2k+1}[u]$, where $\wt \alpha_{2s}^{2k+1}[u] = 2s > 0$, for a given~$ (\alpha_{2s}^{2k+1}[u])'$, the solution, if it exists, is unique in~$\mathcal{H}$. Since such a solution exists according to \cite{BM}, this allows us to uniquely construct the polynomials $\alpha_{2s}^{2k+1}[u] \in \mathcal{H}$ recursively for $s=1,\dots,k $.\qed
\end{proof}

According to Lemma 44 in \cite{BM}, there are differential polynomials $\nu_{2k+2}[u]$ in $u$ such that
\[
[L, A_{2k+1}] = r_{2k+3}[u] = \partial_z (\nu_{2k+2}[u]).
\]
The Korteweg--de Vries hierarchy \eqref{KdVh} turns into an infinite sequence of partial differential equations of the form
\begin{equation} \label{nu}
\frac{\partial u}{\partial t_{2k+1}} = \frac{\partial \nu_{2k+2}[u]}{\partial t_{1}}, \qquad k \in \mathbb{N}.
\end{equation}

\begin{ex}
Explicit expressions for operators $A_{2k+1}$, for $k = 0,1,2,3,4$:
\begin{align*}
A_1 &= \partial_z,\\
A_3 &= \partial_z^3 - \frac34 (u \partial_z + \partial_z u),\\
A_5 &= \partial_z^5 - \frac54 (u \partial_z^3 + \partial_z^3 u) + \frac5{16} ((u'' + 3 u^2) \partial_z + \partial_z (u'' + 3 u^2)),\\
A_7 &= \partial_z^7 - \frac74 (u \partial_z^5 + \partial_z^5 u) + \frac{35}{16} ( \alpha_4^7[u] \partial_z^3 + \partial_z^3 \alpha_4^7[u]) - \frac7{64}( \alpha_6^7[u] \partial_z + \partial_z \alpha_6^7[u]),\\
A_9 &= \partial_z^9 - \frac94 (u \partial_z^7 + \partial_z^7 u) + \frac{21}{16}( \alpha_4^9[u] \partial_z^5 + \partial_z^5 \alpha_4^9[u])\\
& \qquad-\frac{21}{64} ( \alpha_6^9[u] \partial_z^3 + \partial_z^3 \alpha_6^9[u])+ \frac3{256}( \alpha_8^9[u] \partial_z + \partial_z \alpha_8^9[u]),
\end{align*}
where
\begin{gather*}
\alpha_4^7[u] = u''+ u^2, \qquad \alpha_6^7[u] = 13 u^{(4)} + 10 u u'' + 25 u'^2 + 10 u^3, \\
\alpha_4^9[u] = 5 u'' + 3 u^2, \qquad \alpha_6^9[u] = 41 u^{(4)} + 30 u u'' + 55 u'^2 + 10 u^3, \\
\alpha_8^9[u] = 695 u^{(6)} + 546 u u^{(4)} + 2856 u' u''' + 2527 u''^2 + 210 u^2 u'' + 1050 u u'^2 + 105 u^4.
\end{gather*}
\end{ex}

\begin{ex}
Explicit expressions for differential polynomials $r_{2k+3}[u]$, for $k = 0,1,2,3,4$:
\begin{align*}
r_{3}[u] &= u',\\
r_{5}[u] &= \frac14 ( u''' - 6 u u'),\\
r_{7}[u] &= \frac1{16}(u^{(5)} - 10 u u''' - 20 u' u'' + 30 u^2 u'),\\
r_{9}[u] &= \frac1{64} ( u^{(7)} - 14 u u^{(5)} - 42 u' u^{(4)} - 70 u'' u'''\\
&\qquad + 70 u^2 u''' + 280 u u' u'' - 140 u^3 u' + 70 u'^3 ),\\
r_{11}[u] &= \frac1{256} ( u^{(9)} - 18 u u^{(7)} - 72 u' u^{(6)} - 168 u'' u^{(5)} + 126 u^2 u^{(5)} - 252 u''' u^{(4)} + 756 u u' u^{(4)}) + \\
&\qquad + \frac{21}{128} ( 30 u u'' u''' + 23 u'^2 u''' - 10 u^3 u''' + 31 u' u''^2 - 60 u^2 u' u'' - 30 u u'^3 + 15 u^4 u').
\end{align*}
\end{ex}

\section{Commuting Buchstaber--Shorina differential operators} \label{OBS}

Let $g \in \mathbb{N}$. Let also $v_2,\dots, v_{2g}$ be a sequence of differentiable functions of complex variables~$t = (t_1, t_3, t_5,\dots, t_{2g-1})$. Denote by $\partial_i$ the operator of differentiation with respect to the variable~$t_i$. By prime we denote the derivative with respect to $t_1$. We set
\begin{equation}
L = \partial_1^2 - v_2. \label{L}
\end{equation}

We call the following operators
\begin{align}
B_{2k+1} &= \partial_{2k-1} \partial_1^2 - \frac12 (v_2 \partial_{2k-1} + \partial_{2k-1} v_2) - \frac14 (v_{2k} \partial_1 + \partial_1 v_{2k}) - \partial_{2k+1}, \label{U}
\end{align}
where $k = 1,\dots, g$ and $\partial_{2g+1} = 0$, the \emph{Buchstaber--Shorina operators}.

Note that, for $k=1$, $t_1 = z$ and $v_2 = u$, we obtain $B_{3} = U_3$ (see \S \ref{S0}).

\begin{thm}[\rm(see Theorem 4.2 in \cite{BS} and Theorem 1 in \cite{BS1})] \label{TBS}
The following conditions are equivalent:

\begin{itemize}
\item for a sequence of functions $v_2, \dots, v_{2g}$, operators $L, B_{3}, B_{5},\dots,\break B_{2g+1}$, defined by the expressions \eqref{L} and~\eqref{U}, commute;

\item for a sequence of functions $v_2, \dots, v_{2g}$, there exists a set of constants $\mu=(\mu_4, \mu_6,\dots,$ $\mu_{4g+2})$ such that, for generating  functions
\begin{equation} \label{bx}
\mathbf{b}_1(\xi) = \frac12 \sum_{i=1}^g v_{2i} \xi^i, \quad \mathbf{b}_2(\xi) = \frac12 \sum_{i =1}^g v_{2i}' \xi^i, \quad\mathbf{b}_3(\xi) = \frac12 \sum_{i=1}^g v_{2i}'' \xi^i
\end{equation}
and $\mathbf{m}(\xi) = \xi^{-1} + \sum_{i=1}^{2g} \mu_{2i+2} \xi^i$, we have
\begin{equation} \label{mx}
4 \mathbf{m}(\xi) = \mathbf{b}_2(\xi)^2 + 2 \mathbf{b}_3(\xi) (1 - \mathbf{b}_1(\xi)) + 4 (\xi^{-1} + v_2) (1 - \mathbf{b}_1(\xi))^2.
\end{equation}
\end{itemize}

\end{thm}

\begin{thm}[\rm(see Theorems 4.2 and 4.3 in \cite{BS}, Theorems 1 and 2 in \cite{BS1})]\label{TBS2}
Under the conditions of Theorem \ref{TBS}, the function $v_2$ is a solution of the higher stationary $g$-KdV equation~\eqref{gKdV} for $z=t_1$. The constants $\mu_4, \mu_6,\dots, \mu_{2g}$ and the parameters of the higher stationary $g$-KdV equation $a_4, a_6,\dots, a_{2g}$ satisfy the following relations for $k = 1, \dots,g-1$:
\begin{equation} \label{44}
\mu_{2k+2} = 2 a_{2k+2} + \sum_{i=1}^{k-2} a_{2i+2} a_{2k-2i}.
\end{equation}
\end{thm}

In \cite{BS} it is shown that functions $v_{2k}$ from the Theorem \ref{TBS}, for $k=2,3,\dots,g$, are recursively expressed in terms of function $v_2$, its derivatives with respect to $z = t_1$ and constants $\mu_{2i}$, namely (see equation (4.3) in \cite{BS}),
\begin{equation} \label{43}
v_{2k} = T_k(v_2, v_2',\dots, v_2^{(2k-2)}, \mu_{4},\dots, \mu_{2k}),
\end{equation}
where $T_k$ are polynomials. The recursive expression for polynomials $T_k$ can be found using \eqref{mx}. Namely, for $\mathbf{b}_i(\xi)$ given by the equalities \eqref{bx}, we set
\begin{equation} \label{thateps}
\sum_{k=1}^\infty \varepsilon_{2k} \xi^k = \xi \mathbf{b}_2(\xi)^2 + 2 \xi \mathbf{b}_3(\xi) (1 - \mathbf{b}_1(\xi)) + 4 \mathbf{b}_1(\xi)^2 + 4 v_2 \xi (1 - \mathbf{b}_1(\xi))^2.
\end{equation}
We have $\varepsilon_{2} = 4 v_{2}$. For $k \geqslant 2$, equation \eqref{thateps} defines $\varepsilon_{2k}$ as a homogeneous polynomial in $v_{2i}$, $v_{2i}'$, and $v_{2i}''$, where $i \in \{1,\dots,k-1\}$. Relations
\begin{equation} \label{vmu}
v_{2k} = \frac14 \varepsilon_{2k} - \mu_{2k}
\end{equation}
define the required recursive expression.

Note that, for $u = v_2$ and for the constants~$\mu_4, \mu_6,\dots, \mu_{2g}$ given by the relation
\eqref{44}, for $k = 1,\dots,g-1$, we obtain $v_{2k} \in \mathcal{H}_A$ and $\wt v_{2k} = 2 k$.

\begin{lem}[\rm(Lemma 2.1 in \cite{BS} and Lemma 1 in \cite{BS1})]\label{lem33}
Under the conditions of the Theorem \ref{TBS}, for $k = 1,\dots,g-1$, we have
\begin{equation} \label{tpar}
\frac{\partial v_2}{\partial t_{2k+1}} = \frac{\partial v_{2k+2}}{\partial t_{1}}\,.
\end{equation}
\end{lem}

Note that equation \eqref{tpar}, for $k = 1,\dots,g-1$, defines a hierarchy for the function $v_2$, where the right-hand sides of equations of the hierarchy lie in $\mathcal{H}_A$. The higher stationary $g$-KdV equation from Theorem \ref{TBS2} can be considered as the $g$-th equation of this hierarchy.

\begin{lem}[\rm(Theorem 3.1 in \cite{BS})] \label{vppp}
Under the conditions of Theorem \ref{TBS}, for $k = 1,\dots,g$, where $v_{2g+2} = 0$, we have
\[
v_{2k}''' = 4 v_2 v_{2k}' + 2 v_2' v_{2k} + 4 v_{2k+2}'.
\]
\end{lem}

\begin{thm}[\rm(Theorem 4.3 in~\cite{BS} and Theorem 2 in \cite{BS1})] \label{t42}
Under the conditions of Theorem~\ref{TBS}, we have 
\[
\sum_{k=1}^g B_{2k+1} L^{g-k} = A_{2g+1} + \sum_{k=0}^{g-2} a_{2g-2k} A_{2k +1},
\]
where the parameters $a_4, a_6,\dots, a_{2g}$ are related to the constants~$\mu_4, \mu_6,\dots, \mu_{2g}$ via \eqref{44}.
\end{thm}

In~\cite{BS} and \cite{BS2} a problem closely related to the problem of expressing the sigma-fun-\break ction in terms of tau functions was solved (see \cite{Nak-10-AJM}):

\begin{prb}[\rm(problem 1 in \cite{BS})]
Let $v_2$ be the function from Theorem  \ref{TBS}. Consider the following equation for $w$:
\begin{equation} \label{91}
v_2 = - 2 \partial_1^2 \ln w
\end{equation}
with initial conditions $w(0) = 1$, $\partial_k w(0) = 0$, $k=1, 3,\dots, 2g-1$. Find natural additional conditions so that this equation has a unique solution.
\end{prb}

\begin{thm}[\rm(Theorem 9.1 in \cite{BS} and Theorem 1 in \cite{BS2})]\label{T91}
There exists a differen-\break tiable solution $w$ of equation \eqref{91} such that 
\[
v_{2k} = - 2 \partial_1 \partial_{2k-1} \ln w,
\]
for $k=1,\dots, g$, are functions of $v_2, v_4,\dots, v_{2g}$ from Theorem \ref{TBS}.
\end{thm}

The solutions of  equation \eqref{91} described in Theorem  \ref{T91}  are called \emph{special}. In what follows, we will need the following expressions from~\cite{BS}.

Let, as in Theorem \ref{TBS},
$$
\mathbf{b}_1(\xi) = \frac12 \sum_{i=1}^g v_{2i} \xi^i, \quad\mathbf{b}_2(\xi) = \frac12 \sum_{i=1}^g v_{2i}' \xi^i,\quad \mathbf{b}_3(\xi) = \frac12 \sum_{i=1}^g v_{2i}'' \xi^i
$$
and
\begin{align*}
\mathbf{q}(\xi, \eta)&= \mathbf{b}_2(\xi) \mathbf{b}_2(\eta) + (1 - \mathbf{b}_1(\xi)) \mathbf{b}_3(\eta) + \mathbf{b}_3(\xi) (1 - \mathbf{b}_1(\eta))\\
&\kern60pt+ 2 (1 - \mathbf{b}_1(\xi)) (1 - \mathbf{b}_1(\eta)) (\xi^{-1} + \eta^{-1} + 2 v_2),\\
\mathbf{p}(\xi, \eta)&= \frac{\xi^2 \eta^2}{(\xi - \eta)^2} \bigg(2 \xi^{-1} + 2 \eta^{-1} + 4 \sum_{i=1}^g \mu_{4i+2} \xi^i \eta^i\\[-8pt]
&\kern60pt + 2 \sum_{i=0}^{g-1} \mu_{4i+4} (\xi + \eta) \xi^i \eta^i - \mathbf{q}(\xi, \eta)\bigg).
\end{align*}

\begin{lem}[\rm(Corollary 8.1 in \cite{BS})]
Function $\mathbf{p}(\xi, \eta)$ is a polynomial of degree $g$ in the variables $\xi$ and $\eta$.
\end{lem}

We define functions $p_{2i-1,2j-1}$ as coefficients in the expansion
\begin{equation} \label{pij}
\mathbf{p}(\xi, \eta) = \sum_{i=1}^g \sum_{j=1}^g p_{2i-1,2j-1} \xi^i \eta^j.
\end{equation}

\begin{lem}[\rm(Lemma 8.2 in \cite{BS})]
$p_{2i-1,1} = p_{1,2i-1} = v_{2i}$.
\end{lem}

\begin{thm}[\rm(Theorem 9.2 in \cite{BS})] \label{t62}
There is a unique special solution to \eqref{91}, such that
\[
2 \partial_{i} \partial_{j} \ln w = - p_{i,j}
\]
for $i,j \in \{1,3,\dots, 2g-1 \}$.
\end{thm}

\begin{cor} \label{cor4l}
We have $\partial_{1} p_{i,j} = \partial_{i} v_{j+1}$ for $i,j \in \{1,3,\dots, 2g-1 \}$.
\end{cor}

\section{Hyperelliptic functions} \label{Sh}

Let $g \in \mathbb{N}$. Denote the coordinates in the complex space $\mathbb{C}^g$ by $t = (t_1, t_3,\dots, t_{2g-1})$. The vector $\omega \in \mathbb{C}^g$ is called the \textit{period} of a meromorphic function $f$ on $\mathbb{C}^g$, if $f(t+\omega) = f(t)$, for all~$t \in \mathbb{C}^g$. If a meromorphic function $f$ has $2g$ independent periods in~$\mathbb{C}^g$, then it is called an \emph{Abelian function}. Thus, an Abelian function is a meromorphic function on the complex torus $\mathbb{C}^g\!/\Gamma$, where $\Gamma$ is a lattice formed by periods.

A flat non-singular algebraic curve of genus $g$ defines the lattice $\Gamma$ as the set of periods of its holomorphic differentials. The torus $\mathbb{C}^g\!/\Gamma$ is called the \emph{Jacobi manifold} of this curve.

We work with a \emph{universal hyperelliptic curve of genus} $g$ in the model
\begin{equation} \label{Vl}
\mathcal{V}_\lambda = \{(x, y)\in\mathbb{C}^2:
y^2 = x^{2g+1} + \lambda_4 x^{2 g - 1} + \lambda_6 x^{2 g - 2} + \cdots + \lambda_{4 g} x + \lambda_{4 g + 2}\}.
\end{equation}
Each curve is defined by a specialisation of the parameters $\lambda = (\lambda_4, \lambda_6,\dots, \lambda_{4 g}$, $\lambda_{4 g + 2}) \in \mathbb{C}^{2g}$.
We have $\Lambda_g = \mathbb{C}^{2g} \setminus\Sigma$, where $\Sigma$ is the discriminant hypersurface of the universal hyperelliptic curve and $\Lambda_g \subset \mathbb{C}^{2g}$ is a subspace of parameters such that the curve~$\mathcal{V}_{\lambda}$ is non-degenerate for~$\lambda \in \Lambda_g$.

The indices of all variables $t = (t_1, t_3,\dots, t_{2g-1}) \in \mathbb{C}^g$ and parameters $\lambda = (\lambda_4, \lambda_6$, $\ldots, \lambda_{4 g}, \lambda_{4 g + 2}) \in \mathbb{C}^{2 g}$ define their gradings. Namely,
$\wt t_k = - k$ and~$\wt \lambda_k = k$.

For each $\lambda \in \Lambda_g$, the set of periods of the holomorphic differentials on the curve $\mathcal{V}_\lambda$ generates a lattice $\Gamma_\lambda$ of rank $2 g$ in $\mathbb{C}^g$. A  \emph{hyperelliptic function of genus}~$g$ is a meromorphic function on $\mathbb{C}^g \times \Lambda_g$ such that, for each $\lambda \in \Lambda_g$, its restriction to~\text{$\mathbb{C}^g\times\lambda$} is an Abelian function with lattice of periods $\Gamma_\lambda$. Thus, a hyperelliptic function is a function defined on an open dense subset of the total space $\mathcal{U}_g$ of the bundle $\pi\colon \mathcal{U}_g \to \Lambda_g$, whose fiber over $\lambda \in \Lambda_g$ is the Jacobian manifold $\mathcal{J}_\lambda = \mathbb{C}^g\!/\Gamma_\lambda$ of the curve~$\mathcal{V}_\lambda$. A similar bundle of Jacobians of hyperelliptic curves was introduced in~\cite{DN}.

Functions $\wp_{k_1,\dots, k_n}$ given by \eqref{pijdef} are examples of hyperelliptic functions.
We denote by $\mathcal{F}_g$ the subfield of functions, which are represented as rational functions of functions $\wp_{k_1,\dots,k_n}$, of the field of hyperelliptic functions of genus $g$.

\begin{thm}[\rm(Corollary 3.1.2 and Theorem 3.2 in \cite{BEL-97})] \label{t31}
For $i, k \in \{1, 3,\dots, 2 g - 1\}$, we have
\begin{align}
\wp_{1,1,1,i} & = 6 \wp_{1,1} \wp_{1,i} + 6 \wp_{1, i+2} - 2 \wp_{3, i} + 2 \lambda_{4} \delta_{i,1},\label{p111i} \\
\wp_{1,1,i} \wp_{1,1,k} &= 4(\wp_{1,1} \wp_{1,i} \wp_{1,k} + \wp_{1,k} \wp_{1,i+2} + \wp_{1,i} \wp_{1, k+2} + \wp_{k+2, i+2})\nonumber \\
& \qquad - 2 (\wp_{1,i} \wp_{3,k} + \wp_{1,k} \wp_{3,i} + \wp_{k, i+4} + \wp_{i, k+4})\notag \\
& \qquad + 2 \lambda_4 (\delta_{i,1} \wp_{1,k} + \delta_{k,1} \wp_{1,i}) + 2 \lambda_{i+k+4} (2 \delta_{i,k} + \delta_{k, i-2} + \delta_{i, k-2}).
\end{align}
\end{thm}

\section{Commuting Buchstaber--Shorina differential operators for hyperelliptic functions} \label{1BS}

\begin{lem} \label{l51}
For the sequence of functions\begin{equation} \label{posl}
v_2 = 2 \wp_{1,1}, \; v_4 = 2 \wp_{1,3}, \;v_6 = 2 \wp_{1,5}, \;\dots,\; v_{2g}= 2 \wp_{1,2g-1}
\end{equation}
operators $L$ and $B_{2k+1}$, given by \eqref{L} and \eqref{U}, where $k = 1, \dots, g$, commute.
\end{lem}

\begin{proof}
From the equation \eqref{p111i}, for $i = 1$, we have
\[
\wp_{1,1,1,1} = 6 \wp_{1,1}^2 + 4 \wp_{1, 3} + 2 \lambda_{4},
\]
consequently, 
\begin{equation} \label{e1e}
\wp_{1,1,1,1,2k-1} = 12 \wp_{1,1} \wp_{1,1,2k-1} + 4 \wp_{1, 3,2k-1}.
\end{equation}
From the equation  \eqref{p111i}, for $i = 2k-1$, we have
\[
\wp_{1,1,1,2k-1} = 6 \wp_{1,1} \wp_{1,2k-1} + 6 \wp_{1, 2k+1} - 2 \wp_{3, 2k-1} + 2 \lambda_{4} \delta_{2k-1,1},
\]
consequently,
\begin{align}
\wp_{1,1,1,1,2k-1} &= 6 \wp_{1,1,1} \wp_{1,2k-1} + 6 \wp_{1,1} \wp_{1,1,2k-1} + 6 \wp_{1,1, 2k+1} - 2 \wp_{1,3, 2k-1}, \label{e2e}\\
\wp_{1,1,1,2i-1,2k-1} &= 6 \wp_{1,1,2i-1} \wp_{1,2k-1} \nonumber\\
&\qquad+ 6 \wp_{1,1} \wp_{1,2i-1,2k-1} + 6 \wp_{1,2i-1,2k+1} - 2 \wp_{3,2i-1,2k-1}, \nonumber
\end{align}
which imply
\begin{equation} \label{pirel2}
\wp_{1,1,2k-1} \wp_{1,2i-1} - \wp_{1,1,2i-1} \wp_{1,2k-1} + \wp_{1,2k-1,2i+1} - \wp_{1,2i-1,2k+1} = 0.
\end{equation}
From \eqref{e1e} and \eqref{e2e} we obtain
\begin{equation} \label{pirel}
4 \wp_{1,1, 2k+1} =  \wp_{1,1,1,1,2k-1} - 4 \wp_{1,1,1} \wp_{1,2k-1} - 8 \wp_{1,1} \wp_{1,1,2k-1}.
\end{equation}

It follows from Lemma 2.1 in \cite{BS} that $[L,B_{2k+1}] = 0$, for all $k$, if and only if 
\begin{equation}
\partial_1 v_{2k} = \partial_{2k-1} v_2, \quad 4 \partial_{2k+1} v_2 = v_{2k}''' - 2 v_2' v_{2k} - 4 v_2 v_{2k}'. \label{prov}
\end{equation}
Thus, in view of \eqref{pirel}, for the sequence of functions \eqref{posl}
\[
[L,B_{2k+1}] = 0.
\]

According to Lemma 2.4 in \cite{BS}, under the condition \eqref{prov}, the equality $[B_{2i+1},B_{2j+1}] = 0$ holds if and only if for all $i$ , $k$ from $1$ to $g$:
\begin{equation*}
\partial_{2i-1} v_{2j} = \partial_{2j-1} v_{2i}, \qquad v_{2k}' v_{2i} - v_{2i}' v_{2k} + 2 \partial_{2i+1} v_{2k} - 2 \partial_{2k+1} v_{2i}=0.
\end{equation*}
Thus, in view of \eqref{pirel2}, the lemma is proved.
\end{proof}

\begin{cor}[\rm(from Lemma \ref{l51} and Theorem \ref{TBS})] \label{cor53}
There is a set of constants $\mu=(\mu_4, \mu_6,\dots, \mu_{4g+2})$ such that for the generating functions
$$
\mathbf{b}_1(\xi) = \sum_{i=1}^g \wp_{1,2i-1} \xi^i, \quad\mathbf{b}_2(\xi) = \sum_{ i=1}^g \wp_{1,1,2i-1} \xi^i,\quad \mathbf{b}_3(\xi) = \sum_{i=1}^g \wp_{1,1 ,1,2i-1} \xi^i
$$
and~$\mathbf{m}(\xi) = \xi^{-1} + \sum_{i=1}^{2g} \mu_{2i+2} \xi^i$,
\[
4 \mathbf{m}(\xi) = \mathbf{b}_2(\xi)^2 + 2 \mathbf{b}_3(\xi) (1 - \mathbf{b}_1(\xi)) + 4 (\xi^{-1} + 2 \wp_{1,1}) (1 - \mathbf{b}_1(\xi))^2.
\]
\end{cor}

\begin{lem} \label{lemmu}
In Corollary \ref{cor53}
$$
\lambda_4 = \mu_4, \;\lambda_6 = \mu_6,\;\dots,\;\lambda_{4g+2} = \mu_{4g+2}.
$$
\end{lem}

To prove this result, we need the results of \S \ref{s3}, so we present this proof together with the proof of Theorem \ref{t1}.

\begin{cor} We have
\[
\mathcal{V}_\lambda = \{(x,y) \in \mathbb{C}^2 \mid y^2 = x^{2g} \mathbf{m}({1 / x}) \}.
\]
\end{cor}

Using Theorems \ref{TBS2} and \ref{t62}, we obtain the following.

\begin{cor} \label{cor52}
Function $2 \wp_{1,1}$ is a solution of the stationary $g$-KdV equation \eqref{gKdV} with respect to $z = t_1$. The parameters~$\lambda_4, \lambda_6,\dots, \lambda_{2g}$ and the parameters $a_4, a_6,\dots, a_{2g}$ of the stationary $g$-KdV equation satisfy the relations \eqref{44}.
\end{cor}

\begin{ex} The function $u = 2 \wp_{1,1}$ is a solution to equation (see \S\ref{Ssss})
\begin{align*}
&\text{for }\,g=1\quad r_{5}[u]= 0,\\
&\text{for }\,g=2\quad  r_{7}[u] + \frac12 \lambda_4 r_{3}[u] = 0,\\
&\text{for }\,g=3\quad  r_{9}[u] + \frac12 \lambda_4 r_5[u] + \frac12 \lambda_6 r_3[u] = 0,\\
&\text{for }\,g=4\quad  r_{11}[u] + \frac12 \lambda_4 r_7[u] + \frac12 \lambda_6 r_5[u] + \frac12 \bigg(\lambda_8 - \frac14 \lambda_4^2 \bigg) r_3[u] = 0.
\end{align*}
\end{ex}

\begin{cor} \label{cor63}
For the sequence of functions
$$
v_2 = 2 \wp_{1,1},\;v_4 = 2 \wp_{1,3}, \; v_6 = 2 \wp_{1,5}, \;\dots,\; v_{2g} = 2 \wp_{1,2g-1}
$$
in \eqref{pij} we have $p_{i,j} = 2 \wp_{i,j}$.
\end{cor}

\section{Polynomial Mapping Associated with the Universal Jacobian Bundle of Hyperelliptic Curves} \label{s3}

Papers \cite{Prob} and \cite{BEL18} deal with the differentiation of the field of Abelian functions on Jacobi varieties of curves of genus $g$. A particular case of this problem is related to the universal hyperelliptic curve \eqref{Vl}. We develop a new approach to this problem in \cite{B2} and~\cite{B3}. It is based on an algebraic construction that allows to pass from the field $\mathcal{F}_g$ (see \S\ref{Sh}) to the ring of polynomials in the generators of this field. The results of \cite{BS}--\cite{BS2}, presented in \S \ref{OBS} and \S \ref{1BS}, make it possible to obtain an explicit form of the polynomial mapping $\rho_g$ from this construction.

We consider the diagram
\begin{equation} \label{d}
\xymatrix{
\mathcal{U}_g \ar[d]^{\pi} \ar@{-->}[r]^{\varphi} & \mathbb{C}^{3 g} \ar[d]^{\rho_g}\\
\Lambda_g \ar@{^{(}->}[r] & \mathbb{C}^{2g}\\
}
\end{equation}
The bundle $\pi\colon\mathcal{U}_g \to \Lambda_g$ and the embedding $\Lambda_g \subset \mathbb{C}^{2g}$ are described in \S \ref{Sh}.
Theorem \ref{t31} gives a set of relations between the derivatives of functions $\wp_{i,j}$, where~$i,j \in \{1, 3,\dots, 2 g -\nobreak 1\} $, and the parameters $\lambda$. We use it in order to introduce a set of generators into~$\mathcal{F}_g$. The mapping $\varphi$ is defined by this set of generators. The polynomial mapping~$\rho_g$ makes the diagram~\eqref{d} commutative.

The following theorem is a consequence of Theorem \ref{t31}.

\begin{thm}[\rm(Corollary 5.2 in \cite{B3})] \label{c1}
Let $\widetilde{\varphi}\colon \mathcal{U}_g \dashrightarrow \mathbb{C}^{g(g+9)/2}$ be a mapping with coordinates $(b, p, \lambda)$ in $\mathbb{C}^{g(g+9)/2}$, where $b = (b_{i,j}) \in \mathbb{C}^{3g}$, for $i \in \{ 1,2,3 \}$, $j \in \{1, 3,\dots, 2 g -1\}$,
$p = (p_{k,l}) \in \mathbb{C}^{g (g-1)/2}$, for $k,l \in \{3, 5,\dots, 2 g -1\}$, $k \leqslant l$, and $\lambda = (\lambda_s) \in \mathbb{C}^{2 g}$, for~$s \in \{4, 6,\dots, 4 g, 4 g + 2\}$, which is given by
\begin{multline*}
\widetilde{\varphi}\colon(t, \lambda) \mapsto (b_{1,j}, b_{2,j}, b_{3,j}, p_{k,l}, \lambda_s) = (\wp_{1,j}(t, \lambda), \wp_{1,1,j}(t, \lambda), \wp_{1,1,1,j}(t, \lambda), 2 \wp_{k,l}(t, \lambda), \lambda_s),
\end{multline*}
where $p_{l,k} = p_{k,l}$, for $k,l \in \{3, 5,\dots, 2 g -1\}$.
The image of the mapping $\widetilde{\varphi}$ lies on the surface $\mathcal{S}$ defined in $\mathbb{C}^{g (g+9)/2}$ by the system of $g (g+3) /2$ equations
\begin{align} \label{rels}
b_{3,1}&=6 b_{1,1}^2 + 4 b_{1,3} + 2 \lambda_{4},\\
b_{3, k}&=6 b_{1,1} b_{1,k} + 6 b_{1,k+2} - p_{3, k}, \label{rels2} \\
b_{2,1}^2&=4 b_{1,1}^3 + 4 b_{1,1} b_{1, 3} - 4 b_{1, 5} + 2 p_{3, 3} + 4 \lambda_4 b_{1,1} + 4 \lambda_{6}, \label{rels3} \\
b_{2,1} b_{2, k}&= 4 b_{1,1}^2 b_{1, k} + 2 b_{1, 3} b_{1, k} + 4 b_{1,1} b_{1, k+2} - 2 b_{1, k+4}\notag\\
&\qquad- b_{1,1} p_{3,k} + 2 p_{3, k+2} - p_{5, k} + 2 \lambda_4 b_{1,k} + 2 \lambda_{8} \delta_{3,k}, \label{rels4}\\
b_{2,j} b_{2,k}&=4 b_{1,1} b_{1,j} b_{1,k} + 4 b_{1,k} b_{1,j+2} + 4 b_{1,j} b_{1,k+2} + 2 p_{k+2, j+2} - b_{1,j} p_{3,k}\notag\\
&\qquad - b_{1,k} p_{3,j} - p_{k, j+4} - p_{j, k+4} + 2 \lambda_{j+k+4} (2 \delta_{j,k} + \delta_{k, j-2} + \delta_{j, k-2}) \label{rels5}
\end{align}
for $j,k \in \{3,\dots, 2g-1\}$, where $k \geqslant j$ and any variable is zero if its index is outside the specified range.
\end{thm}

\begin{thm}[\rm(Theorem 5.3 in \cite{B3})] \label{thm3}
In Theorem \ref{c1}, the projection of \break$\pi_1\colon \mathbb{C}^{g(g+9)/2} \to \mathbb{C}^{3 g}$ onto the first~$3g$ coordinates defines an isomorphism $\mathcal{S} \simeq \mathbb{C}^{3g}$. Thus, the coordinates~$b = (b_{i,j})$ uniformise~$\mathcal{S}$.
\end{thm}

We have $\widetilde{\varphi}: \mathcal{U}_g \dashrightarrow \mathcal{S} \simeq \mathbb{C}^{3g}$. We denote by $\varphi$ the composition $\pi_1 \circ \widetilde{\varphi}\colon\mathcal{U}_g \dashrightarrow \mathbb{C}^{3g}$.

We obtain the diagram
\begin{equation} \label{diagbig}
\xymatrix{
	 & \mathbb{C}^{g(g+9)/2} \ar@{<-_{)}}[d] \ar@{=}[r] & \mathbb{C}^{3 g} \times \mathbb{C}^{g(g-1)/2} \times \mathbb{C}^{2 g} \ar[dl]^{\pi_1} \ar@/^/[ddl]^{\pi_3}\\
	\mathcal{U}_g \ar[d]^{\pi} \ar@{-->}[r]^(.4){\varphi} \ar@{-->}[ur]^(.5){\widetilde{\varphi}}& \mathcal{S} \simeq \mathbb{C}^{3 g} \ar[d]^{\rho_g}\\
	\Lambda_g \ar@{^{(}->}[r] & \mathbb{C}^{2g}
	}
\end{equation}

\begin{lem}[\rm(Corollary 5.5 in \cite{B3})] \label{cor3}
Projection $\pi_3\colon \mathbb{C}^{g(g+9)/2} \to \mathbb{C}^{2 g}$ onto the last $2 g$ coordinates in Theorem \ref{c1}, bounded on $\mathcal{S} \simeq \mathbb{C}^{3 g}$, defines a polynomial mapping $\rho_g\colon \mathbb{C}^{3g} \to \mathbb{C}^{2g}$ . The coordinates $\lambda = (\lambda_s)$ in $\mathbb{C}^{2 g}$ are expressed as polynomials in the coordinates $b = (b_{i,j})$ in $\mathcal{S}$.
\end{lem}

\begin{lem}[\rm(Corollary 5.4 in \cite{B3})] \label{cor4}
Projection $\pi_2\colon \mathbb{C}^{g(g+9)/2} \to \mathbb{C}^{g(g-1)/2}$ onto the means $g (g-1) /2$ of coordinates in Theorem \ref{c1}, restricted to $\mathcal{S} \simeq \mathbb{C}^{3 g}$, defines a polynomial map $\mathbb{C}^{3g} \to \mathbb{C}^{g(g-1)/2}$. The coordinates $p = (p_{k,l})$ in $\mathbb{C}^{g(g-1)/2}$ are expressed as polynomials in the coordinates $b = (b_{i,j})$ in $\mathcal{S}$.
\end{lem}

\begin{thm} \label{scale}
For $j \in \mathbb{N}$, $j < g$, the next diagram is commutative:
\begin{equation} \label{dnew}
\xymatrix{
	\mathbb{C}^{3 j} \ar[d]^{\rho_{j}} \ar@{^{(}->}[r] & \mathbb{C}^{3g} \ar[d]^{\rho_{g}} \\
	\mathbb{C}^{2j} \ar@{^{(}->}[r] & \mathbb{C}^{2g}\\
	}
\end{equation}
Here, the polynomial mapping $\rho_k$, for $k = j$ and $k = g$, is defined in Lemma \ref{cor3},
the embedding $\mathbb{C}^{3j} \subset \mathbb{C}^{3g}$ is defined by the equality of the $3j$ coordinates in $\mathbb{C}^{3j}$ to the first $3j$ coordinates in $\mathbb {C}^{3g}$, and the remaining $3(g-j)$ coordinates are set equal to zero. The embedding $\mathbb{C}^{2j} \subset \mathbb{C}^{2g}$ is defined by the equality of $2j$ coordinates in~$\mathbb{C}^{2j}$ to the first~$2j$ coordinates in~$\mathbb{C}^{2g}$, and the remaining $2(g-j)$ coordinates are set equal to zero.
\end{thm}

\begin{proof}
Recall the proof of Theorem \ref{thm3}, Lemma \ref{cor3}, and Lemma~\ref{cor4} presented in~\cite{B3}.

In Theorem \ref{c1}, equations \eqref{rels}, \eqref{rels2}, and \eqref{rels3} are equivalent to
\begin{align}
2 \lambda_{4} &= - 6 b_{1,1}^2 + b_{3,1} - 4 b_{1,3}, \nonumber\\
p_{3, k} &= 6 b_{1,1} b_{1,k} - b_{3, k} + 6 b_{1,k+2}, \label{p3} \\
4 \lambda_{6} &= 8 b_{1,1}^3 + b_{2,1}^2 - 2 b_{1,1} b_{3,1} - 8 b_{1,1} b_{1,3} + 2 b_{3, 3} - 8 b_{1,5}, \nonumber
\end{align}
where $k \in \{3, 5,\dots, 2g-1\}$ and $b_{1, 2g+1} = 0$. Recall that $p_{j,k} = p_{k,j}$. Taking these relations into account, equation~\eqref{rels4}, for $k=3$, takes the form
$$
2 \lambda_{8} = 8 b_{1,1}^2 b_{1, 3} - b_{3,1} b_{1,3} + b_{2,1} b_{2,3} - b_{1,1} b_{3,3} + 2 b_{1,3}^2 - 4 b_{1,1} b_{1,5} + b_{3, 5} - 4 b_{1,7}.
$$
For $k \in \{5,7,\dots, 2g-1\}$, equation \eqref{rels4} takes the form
\begin{align}
p_{5, k} &= b_{3,1} b_{1,k} - b_{2,1} b_{2, k} + b_{1,1} b_{3, k} - 2 (4 b_{1,1}^2 + b_{1, 3}) b_{1, k} \nonumber\\
& \qquad + 10 b_{1,1} b_{1, k+2} - 2 b_{3, k+2} + 10 b_{1,k+4}. \label{p5}
\end{align}
For $k \geqslant j+4$, $j, k \in \{3,5,\dots, 2g-1\}$, equation \eqref{rels5} becomes
\begin{align}
p_{j+4,k}&=  2 p_{j+2,k+2} - p_{j, k+4} - 8 b_{1,1} b_{1,j} b_{1,k} - 2 b_{1,k} b_{1,j+2} \nonumber \\
& \qquad- 2 b_{1,j} b_{1,k+2} + b_{1,j} b_{3, k} - b_{2,j} b_{2,k} + b_{3, j} b_{1,k}. \label{pj}
\end{align}
In these equations, any variable is equal to zero if its index is outside the range specified in Theorem  \ref{c1}.
Equations \eqref{p3}, \eqref{p5} and \eqref{pj} express the coordinates $p_{l,k}$, where $l \leqslant k$, $k, l \in \{3,5,\dots, 2g-1\}$, as polynomials in coordinates $b = (b_{i,j})$. Here, we consider equation \eqref{pj} as a recursive expression for $p_{l,k}$, where $l = 7, 9,\dots, 2g-1$.

$j = k \in \{3,5,\dots, 2g-1\}$, equation \eqref{rels5} takes the form
$$
4 \lambda_{2j+4} = b_{2,j} b_{2,j} + 8 b_{1,1} b_{1,j}^2 + 4 b_{1,j} b_{1,j+2}- 2 b_{1,j} b_{3, j} - 2 p_{j+2, j+2} + 2 p_{j,j+4}.
$$
$k = j+2$, $j \in \{3,5,\dots, 2g-3\}$, equation \eqref{rels5} becomes
\begin{align*}
2 \lambda_{2j+6} &= b_{2,j} b_{2,j+2} + 8 b_{1,1} b_{1,j} b_{1,j+2} + 2 b_{1,j+2}^2 + 2 b_{1,j} b_{1,j+4}\\
& \qquad - b_{1,j} b_{3, j+2} - b_{3, j} b_{1,j+2} - p_{j+2,j+4} + p_{j, j+6} .
\end{align*}
In these expressions $b_{i,j} = 0$, for $j > 2g-1$, and $p_{l,k} = 0$, for $l > 2g-1$ or $k > 2g-1$. Taking into account the expressions for $p_{l,k}$, we expressed $\lambda_s$, where~$s \in \{4, 6, 8,\dots, 4g, 4g+2 \}$, as polynomials in coordinates $b = (b_{i,j})$.

The resulting expressions for $p_{k,l}$, where $k \leqslant l$, have the following property: if $b_{i,j} = 0$, for $i \in \{1,2,3\}$ , $j \in \{l, l+2,\dots, 2g-1\}$, then $p_{k,l} = 0$.

The resulting expressions for $\lambda_s$ possess the following property: if $b_{i,j} = 0$, for $i \in \{1,2,3\}$, $j \in \{l, l+2, \dots, 2g-1\}$, then $\lambda_{2l+2} = 0$ and $\lambda_{2l+4} = 0$.

Setting $l = 2j+1$ and $b_{i,k} = 0$, for $i \in \{1,2,3\}$, $k \in \{l, l+2,\dots, 2g-1\}$, we obtain an expression which coincides with the expression for the polynomial mapping~$\rho_j$ (see the lemma \ref{cor3}), which completes the proof of the theorem.\qed
\end{proof}

\begin{rem}\label{rem-1}
Under the embedding $\mathbb{C}^{2j} \subset \mathbb{C}^{2g}$ used in Theorem~\ref{scale}, the image of the manifold $\Lambda_{j}$ does not intersect with the manifold $\Lambda_{g}$. This does not allow one to construct an analogue of diagram (40) for the bundles $\pi\colon\mathcal{U}_g \to \Lambda_g$. Thus, the transition from bundles $\pi$ to mappings $\rho_g \colon \mathbb{C}^{3g}\to \mathbb{C}^{2g}$ is a fundamentally important step in applications of the parametric Korteweg-de Vries hierarchy.
\end{rem}

Let us describe the polynomial mappings from Lemmas \ref{cor3} and \ref{cor4} in terms of the results of \S \ref{OBS} and \S \ref{1BS}.

\begin{thm} \label{t1}
The following relation
\begin{equation} \label{mx2}
4 \mathbf{m}(\xi) = \mathbf{b}_2(\xi)^2 + 2 \mathbf{b}_3(\xi) (1 - \mathbf{b}_1(\xi)) + 4 (\xi^{-1} + 2 b_{1,1}) (1 - \mathbf{b}_1(\xi))^2,
\end{equation}
where
\begin{align}
\label{bm1}
\mathbf{b}_i(\xi) &=\sum_{j=1}^g b_{i,2j-1} \xi^j \quad \text{for } i = 1,2,3, \\
\label{bm2}
\mathbf{m}(\xi) &= \xi^{-1} + \sum_{j=1}^{2g} \lambda_{2j+2} \xi^j,
\end{align}
defines the expressions for the polynomial mapping $\rho_g \colon \mathbb{C}^{3g} \to \mathbb{C}^{2g}$ from Lemma~\ref{cor3}.
\end{thm}

\begin{proof}
The proof follows from Corollary~\ref{cor53}, Lemma \ref{lemmu}, and Theorem~\ref{c1}. It remains for us to prove the result of Lemma \ref{lemmu}, namely prove that
$$
\lambda_4 = \mu_4, \;\lambda_6 = \mu_6, \;\dots, \;\lambda_{4g+2} = \mu_{4g+2}
$$
for
\begin{equation} \label{bm3}
\mathbf{m}(\xi) = \xi^{-1} + \sum_{j=1}^{2g} \mu_{2j+2} \xi^j.
\end{equation}
Note that, since in Corollary~\ref{cor53} the constants $\mu_4, \mu_6,\dots \mu_{4g+2}$ do not depend on the variables $t_1, t_3,\dots, t_{2g-1}$, the expressions \eqref{bm1}, \eqref{mx2} and \eqref{bm3} define $\mu_4, \mu_6,\dots,\mu_{4g+2}$ as functions of $\lambda_4, \lambda_6,\dots,\lambda_{4g+2}$.

We set $b_{1,2j-1} = b_{2,2j-1} = 0$, for $j \in \{1,\dots, g\}$. The expression \eqref{mx2} becomes $2 \mathbf{m}(\xi) - 2 \xi^{-1} = \mathbf{b}_3(\xi)$. In this case, we obtain $2 \mu_{2j+2} = b_{3,2j-1}$, for $j \in \{1,\dots, g\}$. The relations \eqref{rels}--\eqref{rels5} in this case will take the form
\begin{gather*}
b_{3,1} = 2 \lambda_{4}, \quad b_{3, 3} = 2 \lambda_{6}, \quad b_{3, 5} = 2 \lambda_{8},  \\
p_{j,k} = - \frac{(j-1)}2 b_{3, k+j-3},\quad b_{3,2k+1} = 2 \lambda_{2k+4}, \quad b_{3,2k+3} = 2 \lambda_{2k+6}
\end{gather*}
for $j,k \in \{3,\dots, 2g-1\}$, $k \geqslant j$, and any variable is equal to zero if its index is outside the range specified in Theorem~\ref{c1}.

Hence, $\mu_{2j+2} = \lambda_{2j+2}$, for $j \in \{1, \dots, g\}$. The remaining equalities follow from  Theorem \ref{scale}.\qed
\end{proof}

\begin{cor} \label{cor77}
Under the conditions of Theorem \ref{TBS}, the sequence of functions $v_2, v_4,\dots,$ $v_{2g}$ defines a mapping
\[
\psi_v\colon\mathbb{C}^g \dashrightarrow \mathbb{C}^{3g}
\]
from the common domain of functions $v_2, v_4,\dots, v_{2g}$ depending on variables~$t = (t_1, t_3, t_5,\dots$, $t_{2g-1})$ to the space $\mathbb {C}^{3g}$ with coordinates $b_{1,j} = \frac12 v_{j+1}$, $b_{2,j} = \frac12 v_{j+1}'$, $b_{ 3,j} = \frac12 v_{j+1}''$, where $j \in \{1,3,\dots,2g-1\}$.

The image of the composition of mappings $\rho_g \circ \psi_v$ is a point in $\mathbb{C}^{2g}$ with coordinates $\lambda_s = \mu_s$, $s \in \{4,6,\dots, 4g+ 2\}$.
\end{cor}

\begin{thm} \label{t2}
The coefficients of the expansion
\begin{equation} \label{pijn}
\mathbf{p}(\xi, \eta) = \sum_{i=1}^g \sum_{j=1}^g p_{2i-1,2j-1} \xi^i \eta^j,
\end{equation}
where
\begin{align*}
\mathbf{p}(\xi, \eta)&=\frac{\xi^2 \eta^2}{(\xi - \eta)^2} \bigg(2 \xi^{-1} + 2 \eta^{-1} + 4 \sum_{i=1}^g \lambda_{4i+2} \xi^i \eta^i\\[-10pt]
&\hskip60pt + 2 \sum_{i=0}^{g-1} \lambda_{4i+4} (\xi + \eta) \xi^i \eta^i - \mathbf{q}(\xi, \eta)\bigg),\\
\mathbf{q}(\xi, \eta)&= \mathbf{b}_2(\xi) \mathbf{b}_2(\eta) + (1 - \mathbf{b}_1(\xi)) \mathbf{b}_3(\eta) + \mathbf{b}_3(\xi) (1 - \mathbf{b}_1(\eta))\\
&\hskip60pt+ 2 (1 - \mathbf{b}_1(\xi)) (1 - \mathbf{b}_1(\eta)) (\xi^{-1} + \eta^{-1} + 4 b_{1,1})
\end{align*}
and $\mathbf{b}_j(\xi)$ are given by the expressions \eqref{bm1} and the coordinates~$\lambda$ are expressed in terms of the coordinates~$b$ according to Theorem \ref{t1}, give expressions for the polynomial mapping $\mathbb{C}^{3g} \to \mathbb{C}^{g (g-1)/2}$ from Lemma \ref{cor4}.
\end{thm}

The proof follows directly from Corollary~\ref{cor63} and Theorem~\ref{c1}.

\section{Polynomial differential Buchstaber--Shorina operators} \label{s8}

Vector fields $\partial_{k}$, where $k \in \{1,3,\dots, 2g-1\}$, act on the field $\mathcal{F}_g$. In~\cite{B3} polynomial vector fields $\mathcal{D}_k$ in $\mathbb{C}^{3g}$, where $k \in \{1,3,\dots, 2g-1\} $, were defined for all $i \in \{1,2,3\}$ and $j \in \{1,3,\dots,2g-1\}$ by the relations
\begin{equation} \label{keyrel}
\partial_k (\varphi^* b_{i,j}) = \varphi^* \mathcal{D}_k (b_{i,j}).
\end{equation}

\begin{lem}[\rm(Lemma 6.2 in \cite{B3})] \label{ld1}
$$
\mathcal{D}_1 = \sum_{j \in \{1,3,\dots,2g-1\}}\bigg( b_{2,j} \frac\partial{\partial b_{1,j}} + b_{3,j} \frac\partial{\partial b_{2,j}} + 4 (2 b_{1,1} b_{2,j} + b_{2,1} b_{1, j} + b_{2, j+2}) \frac\partial{\partial b_{3,j}}\bigg),
$$
where $b_{2,2 g + 1} = 0$.
\end{lem}

\begin{lem}[\rm(compare with Lemma 6.3 in \cite{B3})] \label{ld2} Set $p_{s,1} = 2 b_{1,s}$. Then,
$$
\mathcal{D}_{s} =
\frac12 \sum_{k \in \{1,3,\dots,2g-1\}}\!\bigg( \mathcal{D}_1(p_{s,k}) \frac\partial{\partial b_{1,k}} + \mathcal{D}_1(\mathcal{D}_1(p_{s,k})) \frac\partial{\partial b_{2,k}}
+ \mathcal{D}_1(\mathcal{D}_1(\mathcal{D}_1(p_{s,k}))) \frac\partial{\partial b_{3,k}} \bigg)
$$
for $s \in \{3, 5,\dots, 2g - 1\}$.
\end{lem}

We denote by $\mathcal{P}$ the ring of polynomials in $\lambda \in \mathbb{C}^{2g}$. We call a vector field $\mathcal{D}$ in $\mathbb{C}^{3g}$ \emph{projectable} for a polynomial mapping~$\rho_g \colon \mathbb{C}^{3g} \to \mathbb{ C}^{2g}$ if there exists a vector field $F$ in $\mathbb{C}^{2g}$ such that $\mathcal{D}(\rho_g^* f) = \rho_g^* F (f)$, for any $f \in \mathcal{P}$. The vector field $F$ is called the \emph{projection} of the vector field $\mathcal{D}$. Therefore, for a projectable vector field $\mathcal{D}$, we have $\mathcal{D}(\rho_g^* \mathcal{P}) \subset \rho_g^* \mathcal{P}$.

\begin{cor} \label{cords} 
For $s \in \{1,3,\dots,2g-1\}$ the polynomial vector fields $\mathcal{D}_s$ are projectable for the polynomial mapping $\rho_g\colon \mathbb{C}^{3g} \to \mathbb{C}^{2g}$. Their projections are zero.
\end{cor}

\begin{cor} \label{cordcom}
For $s \in \{1,3,\dots,2g-1\}$ the polynomial vector fields $\mathcal{D}_s$ commute.
\end{cor}

\begin{thm}
For the mapping $\psi_v$ from Corollary \ref{cor77} we have
\begin{equation} \label{thisrel}
\partial_s (\psi_v^* b_{i,j}) = \psi_v^* \mathcal{D}_s (b_{i,j})
\end{equation}
for all $s \in \{1,3,\dots, 2g-1\}$, $i \in \{1,2,3\}$ and $j \in \{1,3,\dots,2g-1\}$.
\end{thm}

\begin{proof}
For $s = 1$ the relation \eqref{thisrel} follows from Lemma \ref{vppp}.

For $s > 1$ the relation \eqref{thisrel} follows from Corollaries \ref{cor4l} and \ref{cordcom}.\qed
\end{proof}

Set
\begin{equation}
\mathcal{L} = \mathcal{D}_1^2 - 2 b_{1,1}. \label{eL}
\end{equation}

Operators
\begin{align}
\mathcal{B}_{2k+1}&= \mathcal{D}_{2k-1} \mathcal{D}_1^2 - (b_{1,1} \mathcal{D}_{2k-1} + \mathcal{D}_{2k-1} b_{1,1})\notag\\
&\qquad- \frac12 (b_{1,2k-1} \mathcal{D}_1 + \mathcal{D}_1 b_{1,2k-1}) - \mathcal{D}_{2k+1}, \label{eW}
\end{align}
where $k \in \{1,\dots, g\}$ and $\mathcal{D}_{2g+1} = 0$, are called the \emph{Buchstaber--Schorina polynomial differential operators}.

\begin{thm} \label{t95}
The polynomial differential operators $\mathcal{L}$ and $\mathcal{B}_{2k+1}$ in $\mathbb{C}^{3g}$ given by formulas~\eqref{eL} and \eqref{eW}, where $k \in \{1,\dots, g\}$ and $\mathcal{D}_{2g+1} = 0$, commute.
\end{thm}

\begin{proof}
Under the conditions of Lemma \ref{l51}, comparing the expressions~\eqref{L} and~\eqref{U} with \eqref{eL} and~\eqref{eW}, we obtain
$$
L (\varphi^* b_{i,j}) = \varphi^* \mathcal{L} (b_{i,j}), \quad B_{2k+1} (\varphi^* b_{i,j}) = \varphi^* \mathcal{B}_{2k+1} (b_{i,j})
$$
for all $i \in \{1,2,3\}$ and $j \in \{1,3,\dots,2g-1\}$, and, hence, the proof of the Theorem follows from Lemma \ref{l51}.\qed
\end{proof}

The proofs of the following lemmas follow from explicit expressions for the polynomial differential operators $\mathcal{D}_s$, $\mathcal{L}$, and $\mathcal{B}_{2k+1}$.

\begin{lem}
Under the conditions of Theorem \ref{scale}, for $s \in \{1,3,\dots,$ $2j-1\}$, the vector fields $\mathcal{D}_s$ defined for genus $g$, subject to the constraint on~$\mathbb{C}^{3j}$, coincide with the corresponding vector fields defined for genus $j$.
\end{lem}

\begin{lem}
Under the conditions of Theorem \ref{scale}, for $s \in \{2j+1, 2j+3,\dots,2g-1\}$,
the vector fields $\mathcal{D}_s$ defined for genus $g$ become zero when restricted to $\mathbb{C}^{3j}$.
\end{lem}

\begin{lem} \label{rem9}
Under the conditions of Theorem \ref{scale}, for $k \in \{1,\dots,j\}$, the polynomial differential operators $\mathcal{L}$ and $\mathcal{B}_{2k+1}$ defined for genus $g$, when restricted to~$\mathbb{C}^{3j}$, coincide with the corresponding polynomial differential operators defined for genus $j$.
\end{lem}

\begin{lem} \label{rem9b}
Under the conditions of Theorem \ref{scale}, for $k \in \{j+1, j+2,\dots,g\}$, the
polynomial differential operators $\mathcal{B}_{2k+1}$ defined for genus~$g$ become zero when restricted to~$\mathbb{C}^{3j}$.
\end{lem}

\section{Polynomial parametric Korteweg--de Vries hierarchy} \label{Sinf}

\begin{lem} \label{lem1}
The following diagram is commutative
\begin{equation} \label{dinf}
\xymatrix{
\mathbb{C}^{3 g} \ar@{^{(}->}[r] \ar[d]^{\rho_g} & \mathbb{C}^{3\infty} \ar[d]^{\rho}\\
\mathbb{C}^{2g} \ar@{^{(}->}[r] & \mathbb{C}^{2\infty}\\
}
\end{equation}
Here, $\mathbb{C}^{3\infty}$ and $\mathbb{C}^{2\infty}$ denote infinite-dimensional complex spaces with coordinates $b = (b_{1,2j-1}, b_{2 ,2j-1}, b_{3,2j-1})$ and $\lambda = (\lambda_{4j}, \lambda_{4j+2})$ for $j \in \mathbb{N}$, the embedding $\mathbb{C}^{3g} \subset \mathbb{C}^{3\infty}$ is defined by the equality of the $3g$ coordinates in $\mathbb{C}^{3g}$ to the first $3g$ coordinates in $ \mathbb{C}^{3\infty}$ and the remaining coordinates are assumed to be zero, the embedding $\mathbb{C}^{2g} \subset \mathbb{C}^{2\infty}$ is defined by the equality of $2g$ coordinates in $\mathbb{C}^{2g}$ to the first $2g$ coordinates in $\mathbb{C}^{2\infty}$ and the remaining coordinates are assumed to be zero, the polynomial mapping $\rho_g$ is defined in Lemma \ref{cor3}, and the polynomial mapping $\rho \colon \mathbb{C}^{3\infty} \to \mathbb{C}^{2 \infty}$ is given by \eqref{mx2}, where
\begin{equation} \label{bminf}
\mathbf{b}_i(\xi) =\sum_{j=1}^\infty b_{i,2j-1} \xi^j \quad \text{for }\,i = 1,2,3, \quad \text{and } \quad
\mathbf{m}(\xi) = \xi^{-1} + \sum_{j=1}^\infty \lambda_{2j+2} \xi^j.
\end{equation}
\end{lem}

The proof of the lemma follows from Theorems \ref{scale} and \ref{t1}.

Note that due to the grading, each coordinate $\lambda_{2k+2}$ in \eqref{mx2} is expressed as a polynomial in a finite number of variables, namely the coordinates $b_{i,2j-1}$, such that $i +2j-1 \leqslant 2k+2$.

We denote by $\varphi_g$ the composition of the mapping $\varphi$ from diagram \eqref{d} and the embedding $\mathbb{C}^{3g} \subset \mathbb{C}^{3\infty}$ from diagram \eqref{dinf}. Lemma \ref{lem1} implies that the following diagram
\begin{equation}
\xymatrix{
\mathcal{U}_g \ar[d]^{\pi} \ar@{-->}[r]^{\varphi_g} & \mathbb{C}^{3\infty} \ar[d]^{\rho}\\
\Lambda_g \ar@{^{(}->}[r] & \mathbb{C}^{2\infty}\\
}
\end{equation}
is commutative. Note that the right side of the diagram does not depend on genus $g \in \mathbb{N}$.

\begin{cor}[\rm(from Lemma \ref{ld1} and Corollary \ref{cords})] The polynomial vector field
\begin{equation} \label{d1c}
\mathcal{D}_1 = \sum_{j \in \{1,3,\dots,2g-1,\dots\}}\bigg(b_{2,j} \frac\partial{\partial b_{1,j}} + b_{3,j}\frac\partial{\partial b_{2,j}} + 4 (2 b_{1,1} b_{2,j} + b_{2,1} b_{1,j} + b_{2,j+2}) \frac\partial{\partial b_{3,j}}\bigg)
\end{equation}
in~$\mathbb{C}^{3\infty}$ is projectable for $\rho$. Its projection is zero. Let $\varphi_g^*b_{2,2g+1}=0$. The relation
\[
\partial_1 (\varphi_g^*b_{i,j}) = \varphi_g^* \mathcal{D}_1 (b_{i,j})
\]
is satisfied for all $g \in \mathbb{N}$, $i \in \{1,2,3\}$ and $j \in \{1,3,\dots,2g-1\}$.
\end{cor}

Let $p_{s,k}$ for odd $s , k \geqslant 1$ be polynomials in~$\mathbb{C}^{3\infty}$ defined by \eqref{pijn}, where $\mathbf{b}_j(\xi)$ are given by \eqref{bminf}. We have $p_{s,1} = 2 b_{1,s}$. Note that due to the grading, each coefficient $p_{s,k}$ is expressed as a polynomial in a finite number of variables, namely the coordinates $b_{i,j}$, such that $i+j \leqslant s+k$.

\begin{cor}[\rm(from Lemma \ref{ld2} and Corollary \ref{cords})] For odd $s \geqslant 3$, the polynomial vector fields in~$\mathbb{C}^{3\infty}$
\begin{equation} \label{dsc}
\mathcal{D}_{s} =
\frac12\!\sum_{j \in \{1,3,\dots,2g-1,\dots\}}\!\!\!\bigg( \mathcal{D}_1(p_{s,j})\frac\partial{\partial b_{1,j}} + \mathcal{D}_1(\mathcal{D}_1(p_{s,j})) \frac\partial{\partial b_{2,j}}+ \mathcal{D}_1(\mathcal{D}_1(\mathcal{D}_1(p_{s,j}))) \frac\partial{\partial b_{3,j}} \bigg)
\end{equation}
are projectable for $\rho$. Their projections are zero. We set $\varphi_g^*b_{i,j}=0$, for $j\ge 2g+1$. The relation
\[
\partial_s (\varphi_g^* b_{i,j}) = \varphi_g^* \mathcal{D}_s (b_{i,j})
\]
is satisfied for all $g \in \mathbb{N}$, $i \in \{1,2,3\}$ and $j \in \{1,3,\dots,2g-1\}$.
\end{cor}

\begin{cor}[\rm(from Corollary \ref{cordcom})] \label{cordcoms}
For $s \in \{1,3,\dots,2g-1,\dots\}$  the polynomial vector fields $\mathcal{D}_s$ commute.
\end{cor}

Similarly to \S \ref{s8}, we set
$$
\mathcal{L} = \mathcal{D}_1^2 - 2 b_{1,1}.
$$

In the infinite-dimensional case by \emph{Buchstaber--Shorina polynomial differential operators} we mean the operators
$$
\mathcal{B}_{2k+1} = \mathcal{D}_{2k-1} \mathcal{D}_1^2 - (b_{1,1} \mathcal{D}_{2k-1} + \mathcal{D}_{2k-1} b_{1,1})
- \frac12 (b_{1,2k-1} \mathcal{D}_1 + \mathcal{D}_1 b_{1,2k-1}) - \mathcal{D}_{2k+1}
$$
in $\mathbb{C}^{3\infty}$, where $k \in \{1,\dots, g,\dots\}$.

\begin{cor}[\rm(from Theorem \ref{t95})] \label{c95}
The polynomial differential operators $\mathcal{L}$ and $\mathcal{B}_{2k+1}$ in $\mathbb{C}^{3\infty}$, where~$k \in \{1,\dots, g,\dots\}$, commute.
\end{cor}

We define an infinite set of parameters $a_4, a_6,\dots, a_{2g},\dots$ as polynomials in coordinates $\lambda$ recursively for $k = 1,\dots,g-1,\dots$ via the relations
\begin{equation} \label{par}
2 a_{2k+2} = \lambda_{2k+2} - \sum_{i=1}^{k-2} a_{2i+2} a_{2k-2i}.
\end{equation}
Note that, for $\lambda_{2k+2} = \mu_{2k+2}$, we obtain relations \eqref{44}.

We define polynomial differential operators $\mathcal{A}_{2k+1}$ in $\mathbb{C}^{3 \infty}$ recursively for $k \in\{0,1,\dots$, $g,\dots \}$ using the relation
\begin{equation} \label{relA}
\mathcal{A}_{2k+1} = \sum_{j=1}^k \mathcal{B}_{2j+1} \mathcal{L}^{k-j} + \mathcal{D}_{2k+1} - \sum_{j=0}^{k-2} a_{2k-2j} \mathcal{A}_{2j+1}.
\end{equation}

\begin{thm} \label{lemcom} For all $k \in \mathbb{N}$, we have
\begin{equation} \label{KdPol}
\mathcal{D}_{2k+1}(2 b_{1,1}) =\bigg[\mathcal{L}, \mathcal{A}_{2k+1} + \sum_{j=0}^{k-2} a_{2k-2j} \mathcal{A}_{2j+1}\bigg].
\end{equation}
\end{thm}
\begin{proof}
According to \eqref{relA} and Corollaries \ref{c95} and \ref{cordcoms}, we obtain
\begin{align*}
\kern40pt[\mathcal{L}, \mathcal{A}_{2k+1} + \sum_{j=0}^{k-2} a_{2k-2j} \mathcal{A}_{2j+1}]&= [\mathcal{L}, \sum_{j=1}^k \mathcal{B}_{2j+1} \mathcal{L}^{k-j} + \mathcal{D}_{2k+1}]\\
&= [\mathcal{L}, \mathcal{D}_{2k+1}] = \mathcal{D}_{2k+1}(2 b_{1,1}).\kern60pt\Box
\end{align*}
\end{proof}

We call the system of equations \eqref{KdPol}, for $k \in \mathbb{N}$, the \emph{polynomial parametric Korteweg--de Vries  hierarchy}. This concept is based on the following results.

\begin{thm} \label{tmain}
For all $g \in \mathbb{N}$, $i \in \{1,2,3\}$ and $j \in \{1,3,\dots,2g-1\}$ we have
\begin{equation} \label{keyrelA}
A_{2k+1} (\varphi_g^* b_{i,j}) = \varphi_g^* \mathcal{A}_{2k+1} (b_{i,j}).
\end{equation}
\end{thm}

\begin{proof}
The relation \eqref{relA} defines a polynomial differential operator~$\mathcal{A}_{2k+1}$ as an operator in $\mathcal{D}_{1}, \mathcal{D}_{3},\dots , \mathcal{D}_{2k-1}$ with coefficients in $b_{i,j}$, where $j \leqslant 2k-1$. Thus, if \eqref{keyrelA} holds for~$g = k$, then it also holds for $k<g$.

According to Lemma \ref{rem9}, if \eqref{keyrelA} holds for~$g = k$, then it also holds for $k>g$.

For $g=k$, the relation \eqref{keyrelA} follows from Theorem \ref{t42} for functions \eqref{posl}.\qed
\end{proof}

\begin{cor} \label{thatcor}
For mappings $\psi_v$ from Corollary \ref{cor77}, we have 
\begin{equation}
A_{2k+1} (\psi_v^* b_{i,j}) = \psi_v^* \mathcal{A}_{2k+1} (b_{i,j}).
\end{equation}
where $i \in \{1,2,3\}$ and $j \in \{1,3,\dots,2g-1\}$.
\end{cor}

From Theorems~\ref{lemcom} and~\ref{tmain}, and Corollary~\ref{thatcor}, we obtain the following results.

\begin{thm} \label{fullhd}
Under the conditions of Theorem \ref{TBS}, function $v_2$ is a solution to the parametric Korteweg--de Vries hierarchy \eqref{KdVhld}. The parameters $a_4, a_6,\dots$, $a_{2k},\dots$ of the Korteweg--de Vries hierarchy are given by \eqref{par}, for $k = 1,2,3,\dots$, where $\lambda_ {2k+2} = \mu_{2k+2}$ for $k \leqslant 2g$ and $\lambda_{2k+2} = 0$ for $k > 2g$.
\end{thm}

Comparing this result with the result of the lemma \ref{lem33}, we obtain a representation of the parametric Korteweg--de Vries hierarchy \eqref{KdVhld} as \eqref{nu}, where $\nu_{2k+2}[u]\break\in \mathcal{H}_A$.

\begin{cor} \label{lastcor}
In the case $\mu_4 = \mu_6 = \ldots = \mu_{4g+2} = 0$, the expressions \eqref{43}, for $k \leqslant g$, define $v_{2k}$ as differential polynomials from $v_2$ to $z=t_1$, therefore
\[
r_{2k+3}[v_2] = \partial_z (v_{2k+2}).
\]
\end{cor}

\begin{thm} \label{mainthm}
For each $g\in \mathbb{N}$, function
\[
\wp_{1,1} = - \partial_1^2 \ln \sigma(t, \lambda)
\]
(see \eqref{e1}) defines a solution $u = 2 \wp_{1,1}$ to the parametric Korteweg--de Vries hierarchy \eqref{KdVhld}. The parameters $a_4, a_6,\dots, a_{2k},\dots$ of the Korteweg--de Vries hierarchy are related to the parameters $\lambda = (\lambda_4, \lambda_6,\dots, \lambda_{4g+2})$ of the sigma-function by relations \eqref{par}, for $k = 1,2,\dots$, where $\lambda_{2k+2} = 0$, for $k>2g$.
\end{thm}

As a result, we obtain the following generalisation of Theorem \ref{fullhd}.

\begin{thm} \label{tgreat}
Let $v_2, v_4,\dots, v_{2g},\dots$ be an infinite sequence of differentiable functions in an infinite set of complex variables $t = (t_1, t_3, t_5,\dots, t_{2g-1},\dots)$. We denote by a prime the derivative with respect to $t_1$. We define operators $L$ and $B_{2k+1}$ by \eqref{L} and~\eqref{U} respectively, for $k = 1,2,\dots,g,\dots$\,. The following conditions are equivalent:

\begin{itemize}
\item operators $L, B_{3}, B_{5},\dots, B_{2g+1},\dots$ commute;
 
\item the following hold
\[
\frac{\partial v_2}{\partial t_{2k+1}} = \frac{\partial v_{2k+2}}{\partial t_{1}}, \qquad k \in \mathbb{N},
\]
and there is a set of constants $\mu=(\mu_4, \mu_6,\dots, \mu_{4g+2},\dots)$ such that, for generating functions
$$
\mathbf{b}_1(\xi) = \frac12 \sum_{i=1}^\infty v_{2i} \xi^i,\quad\mathbf{b}_2(\xi) = \frac12 \sum_{i=1}^\infty v_{2i}' \xi^i, \quad\mathbf{b}_3(\xi)= \frac12 \sum_{i=1}^\infty v_{2i}'' \xi^i
$$
and $\mathbf{m}(\xi) = \xi^{-1} + \sum_{i=1}^{\infty} \mu_{2i+2} \xi^i$, relation \eqref{mx} holds;

\item  function $u=v_2$ is a solution to the parametric Korteweg--de Vries hierarchy \eqref{KdVhld} with some set of parameters $a_4,a_6,\dots,a_{2k},\dots$, and  functions $v_4, v_6 , \dots, v_{2g}, \dots$ are given by \eqref{43}, where the constants $\mu_{2k+2}$, for $k \in \mathbb{N}$, are given by 
\begin{equation} \label{lastrel}
\mu_{2k+2} = 2 a_{2k+2} + \sum_{i=1}^{k-2} a_{2i+2} a_{2k-2i}.
\end{equation}
\end{itemize}
\end{thm}

\begin{proof}
The equivalence of conditions (i) and (ii) follows from Theorem \ref{TBS} and Lemma \ref{lem33}. Indeed, since all equations are homogeneous with respect to the grading of variables introduced in this paper, the relations on the functions $v_2, v_4,\dots, v_{2g},\dots$ of weight $N$ that arise in the proof of this equivalence coincide with similar relations from Theorems \ref{TBS} for $g>N$.

The equivalence of conditions (ii) and (iii) follows from the fact that, under the conditions of the theorem, relations \eqref{43} are by construction equivalent to \eqref{mx} for the same constants $\mu_{2k+2}$. Moreover, the right parts of relations \eqref{tpar} correspond to the right parts of the hierarchy \eqref{KdVhld} in view of Theorem \ref{fullhd}. Relations \eqref{lastrel} follow from the \eqref{par} relations for~$k\in \mathbb{N}$ in the case $\lambda_{2k+2} = \mu_{2k+2}$.\qed
\end{proof}


\end{document}